\documentclass[11pt]{article}
\usepackage{booktabs}
\usepackage[utf8]{inputenc}

\usepackage{amsmath,amsfonts,amsthm,mathrsfs,xspace,graphicx}
\usepackage[backref,colorlinks,bookmarks=true, citecolor=blue, urlcolor=black]{hyperref}
\usepackage{endnotes}
\usepackage{color}
\usepackage{float}
\usepackage{xcolor}
\definecolor{since}{rgb}{0.5,0.5,0.5}
\definecolor{newred}{HTML}{ED2024}
\definecolor{newgreen}{HTML}{109A48}
\definecolor{newblue}{HTML}{535DAA}
\definecolor{neworange}{HTML}{F79420}
\usepackage{mdframed}
\usepackage{bbm}
\usepackage{suffix} \usepackage{times}
\usepackage{tabularx}
\usepackage{makecell}
\usepackage{amssymb,latexsym}
\usepackage[capitalize]{cleveref}
\usepackage[font=footnotesize,labelfont=bf]{caption}
\usepackage{enumitem}
\usepackage{tikz}
\usepackage{tikz-cd}
\usepackage{multirow}

\usepackage[section]{placeins}
\usepackage[affil-it]{authblk}

\makeatletter
\renewcommand*\env@matrix[1][*\c@MaxMatrixCols c]{%
  \hskip -\arraycolsep
  \let\@ifnextchar\new@ifnextchar
  \array{#1}}
\makeatother

\usepackage[margin=1in]{geometry}

\usepackage{endnotes}
\usepackage{amssymb,latexsym}
\usepackage{enumitem}
\setlist{itemsep=0mm}
\usepackage{tikz}
\usepackage{float}
\usepackage{setspace}
\usepackage{soul}

\usepackage{thmtools}
\usepackage{thm-restate}

\usepackage{mathtools, physics}
\usepackage{stmaryrd}
\usepackage{complexity}

\newclass{\QPCP}{QPCP}
\newclass{\QCPCP}{QCPCP}
\newclass{\QCMAcomp}{QCMA-complete}
\newclass{\sharpP}{\#P}

\newcommand{\defeq}{\vcentcolon=}

\usepackage[linesnumbered,ruled,vlined]{algorithm2e}
\usepackage{verbatim}

\usepackage{amsthm}
\newtheorem{theorem}{Theorem}
\newtheorem*{theorem*}{Theorem}

\newtheorem*{proposition*}{Proposition}
\newtheorem{fact}[theorem]{Fact}
\newtheorem*{fact*}{Fact}
\newtheorem{lemma}[theorem]{Lemma}
\newtheorem*{lemma*}{Lemma}
\newtheorem{claim}[theorem]{Claim}

\newtheorem*{conjecture*}{Conjecture}

\theoremstyle{definition}
\newtheorem{definition}[theorem]{Definition}
\newtheorem*{definition*}{Definition}

\theoremstyle{remark}

\newtheorem*{remark*}{Remark}

\newcommand{\mcH}{\mathcal{H}}

\DeclareMathOperator\eye{\mathbb{I}}

\usepackage{accents}

\newcommand\restr[2]{{
  \left.\kern-\nulldelimiterspace 
  #1 
  \right|_{#2} 
  }}

\usepackage{subcaption}

\newcommand{\jnote}[1]{}

\newcommand*\centertile{\includegraphics[width=.8em]{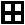}}
\newcommand*\horizontaltile{\includegraphics[width=.8em]{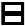}}

\newcommand*\bottomtile{\includegraphics[width=.8em]{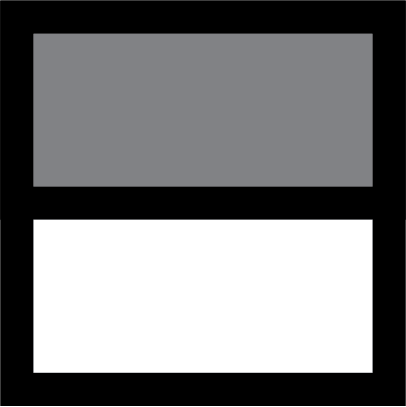}}
\newcommand*\bottomlefttile{\includegraphics[width=.8em]{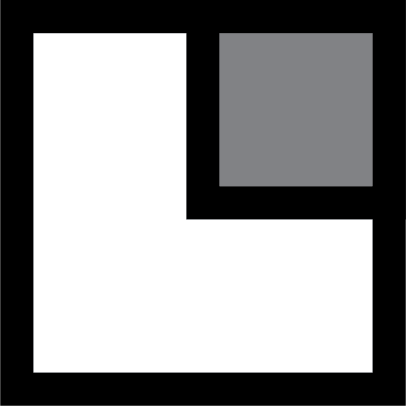}}
\newcommand*\bottomrighttile{\includegraphics[width=.8em]{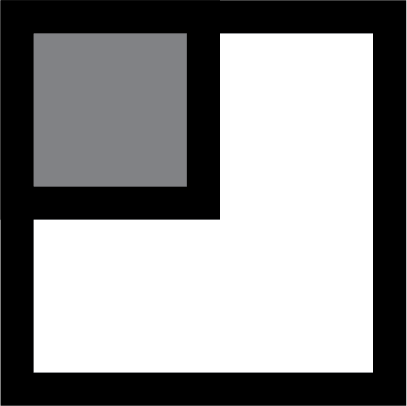}}
\newcommand*\middletile{\includegraphics[width=.8em]{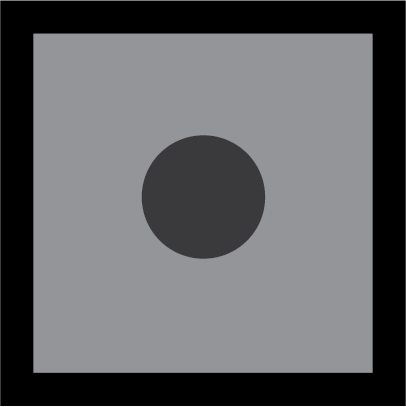}}
\newcommand*\lefttile{\includegraphics[width=.8em]{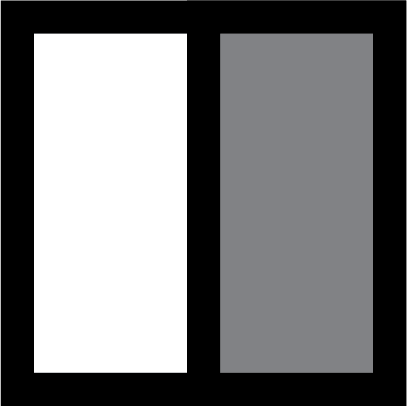}}
\newcommand*\righttile{\includegraphics[width=.8em]{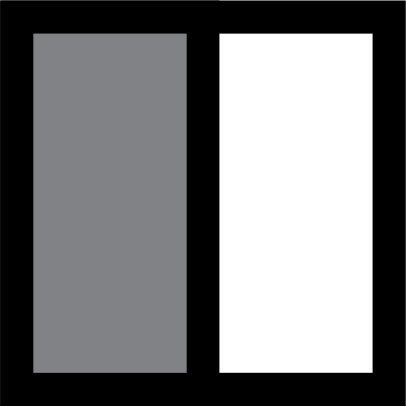}}
\newcommand*\toptile{\includegraphics[width=.8em]{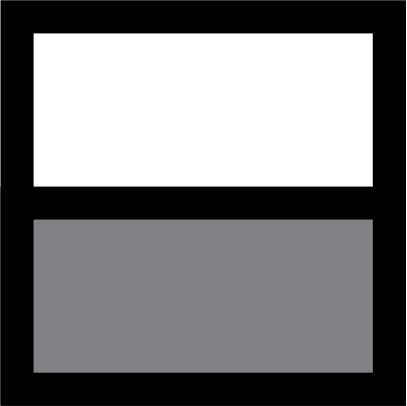}}
\newcommand*\bottombc{\includegraphics[width=.8em]{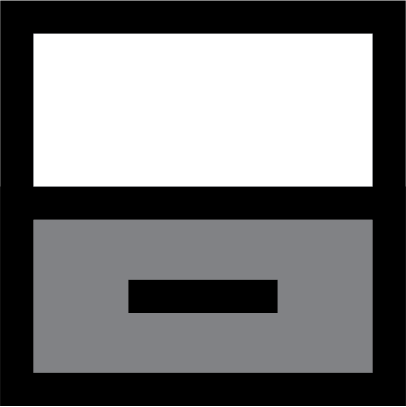}}
\newcommand*\leftbc{\includegraphics[width=.8em]{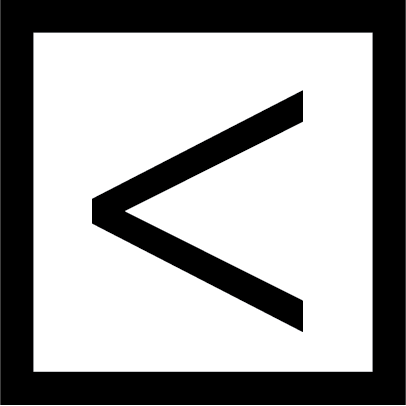}}
\newcommand*\rightbc{\includegraphics[width=.8em]{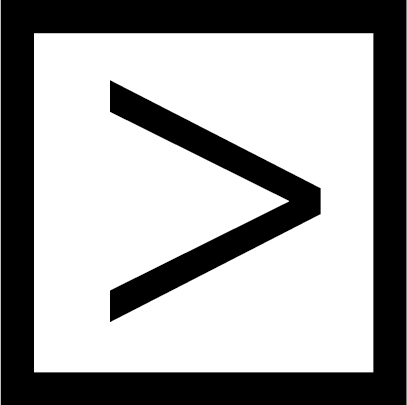}}
\newcommand*\blank{\includegraphics[width=.8em]{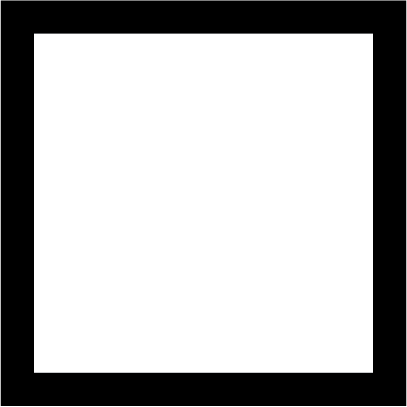}}

\newcommand*\toplefttile{\includegraphics[width=.8em]{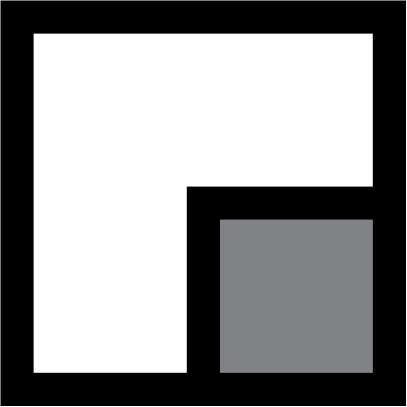}}
\newcommand*\toprighttile{\includegraphics[width=.8em]{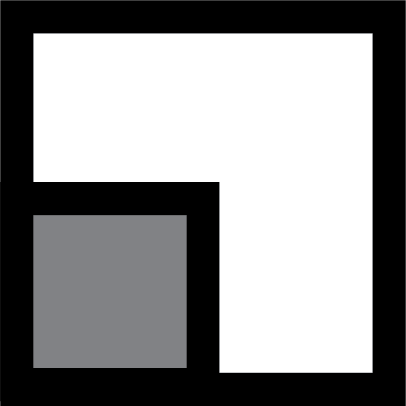}}

\newcommand{\qmaexp}{\text{QMA}_{\text{EXP}}}
\newcommand{\qcmaexp}{\text{QCMA}_{\text{EXP}}}
\title{\vspace{-2.5em}The rotation-invariant Hamiltonian problem is $\qmaexp$-complete}

\author[1]{Jon Nelson \footnote{\href{mailto:nelson1@umd.edu}{nelson1@umd.edu}}}
\author[1]{Daniel Gottesman}

\affil[1]{Joint Center for Quantum Information and Computer Science (QuICS), University of Maryland}
\affil{Department of Computer Science, University of Maryland}

\date{}

\begin{document}

\maketitle

\vspace{-3em}
\begin{abstract}
In this work we study a variant of the local Hamiltonian problem where we restrict to Hamiltonians that live on a lattice and are invariant under translations and rotations of the lattice. In the one-dimensional case this problem is known to be $\qmaexp$-complete. On the other hand, if we fix the lattice length then in the high-dimensional limit the ground state becomes unentangled due to arguments from mean-field theory. We take steps towards understanding this complexity spectrum by studying a problem that is intermediate between these two extremes. Namely, we consider the regime where the lattice dimension is arbitrary but fixed and the lattice length is scaled. We prove that this rotation-invariant Hamiltonian problem is $\qmaexp$-complete answering an open question of \cite{v009a002}. This characterizes a broad parameter range in which these rotation-invariant Hamiltonians have high computational complexity.
\end{abstract}

\section{Introduction}
\label{sec:intro}
In order to understand the behavior of a quantum many-body system, it is crucial to study its Hamiltonian. The Hamiltonian operator not only governs the system's dynamics through the Schr{\"o}dinger equation but also encodes its low-energy states and energy spectrum. It is thus important to understand which Hamiltonians are tractable to analyze. In this work, we study the computational complexity of estimating the ground-state energy of a Hamiltonian with only short-range interactions, which is known as the local Hamiltonian problem. 

Often the goal is to show that a specific variant of this problem is QMA-complete, which implies that for certain Hamiltonians not even a quantum computer can be expected to find its ground state energy. On the one hand, this is a negative result for being able to calculate ground state energies. On the other hand, these results lead to constructions of highly complex quantum systems that are interesting objects of study in their own right. 

Kitaev initiated the study of local Hamiltonian problems in his landmark result proving that this problem in its most general form is QMA-complete \cite{kitaev}. However, Kitaev's result only applies for a worst-case family of Hamiltonians, which are not physically natural. In order to study the complexity of more physically relevant cases, subsequent work has extended Kitaev's result to apply under additional constraints that capture what it means to be ``natural''. This has included restricting the local dimension of each particle as well as constraining the geometry of the interactions. For example, \cite{oliveira} extends Kitaev's result to qubits on a 2D lattice while \cite{Aharonov_2009} further restricts to particles on a line with constant local dimension. 

Additional follow-up work has also emphasized symmetry constraints. This is motivated by the observation that many systems in nature are highly symmetric. For instance, the laws of gravity, electromagnetism, etc., do not change depending on where you are or how you are oriented; thus, these laws are translation and rotation-invariant. For translation-invariant one-dimensional spin chains, \cite{v009a002} showed that this local Hamiltonian problem is $\qmaexp$-complete.

Although much work is now known about the complexity of translation-invariant systems \cite{v009a002,Bausch_2017,Bausch_2021,2019JSP...176..228K,piddock2020universaltranslationallyinvarianthamiltonians}, there have been very few results for the rotation-invariant case. In fact, it was posed as an open question of \cite{v009a002} whether their results can be extended to rotation-invariant Hamiltonians in higher dimensions. In this work, we solve this question and hope similar techniques can be used to lift other translation-invariant results to the rotation-invariant case.

Translation invariance with reflection symmetry and rotation invariance coincide in 1D and so the main challenge is to extend \cite{v009a002} to Hamiltonians on higher dimensional lattices. It is trivial to extend their result to higher dimensional translation-invariant lattices simply by ignoring all but one dimension. However, this breaks the rotation symmetry, which requires that the Hamiltonian terms act identically in all directions. In this case, the key challenge is to handle the increasingly high degree of interaction without increasing the number of parameters in the Hamiltonian. This presents issues, for example, when attempting to encode computation into the Hamiltonian's ground state, which is an essential step for proving hardness. Controlling this computation requires the ability to track time, which can be accomplished in the 1D setting by moving a clock pointer along the spin chain \cite{v009a002}. However, this same idea cannot be used in higher dimensions since the paths can branch in many directions throughout the lattice. In order to pick out a specific time direction, we must engineer a family of Hamiltonians that spontaneously breaks the rotation symmetry.

Technical difficulties aside, it may seem intuitive that increasing the lattice dimension only makes the local Hamiltonian problem more difficult, and so one might assume that the complexity for higher dimensional cases follows from the one-dimensional case. However, due to the rotation symmetry, the increase in lattice dimension does not correspond to more Hamiltonian parameters, and so it is unclear how the complexity actually compares. In fact, standard condensed matter arguments imply that increasing the dimension can instead make the problem easier. This follows from the observation that for higher-dimensional lattices, mean-field theory (which uses a product-state ansatz to approximate the ground state) becomes more and more accurate \cite{https://doi.org/10.1002/andp.19945060606}. In the quantum setting, this can be explained by an effect called monogamy of entanglement \cite{Terhal_2004}, which states that a particle cannot be highly entangled with many other particles. Thus, for high lattice dimension, each particle has many neighbors and so on average they must be nearly unentangled. Due to this effect, the product state becomes a good approximation of the ground state, suggesting that this problem could now be more tractable than the lower dimensional cases.

This has been formalized in \cite{brandaoharrow} which shows that for lattice dimension $r$ there is a product state that approximates the ground-state energy by an average error of $O(r^{-1/3})$ per term of the Hamiltonian. Furthermore, \cite{Kraus_2013} rigorously show that in the limit as $r \rightarrow \infty$ the ground state is exactly a product state when the Hamiltonian is translation and rotation invariant. These results suggest that if the lattice dimension is high enough, the problem loses its quantum hardness since the low-energy states become unentangled. Another result that captures this phenomenon is \cite{aharonov2013commutinglocalhamiltoniansexpanders}, who show that a commuting version of the local Hamiltonian problem becomes easier as the interaction graphs become more expanding, which intuitively corresponds to more interaction.

In this work, we consider a lattice dimension that is in an intermediate regime between one-dimensional spin chains, which are hard and spin chains with $r \rightarrow \infty$, which are easy. In particular, we consider an arbitrary but fixed lattice dimension and show that this rotation-invariant Hamiltonian problem is quantumly hard as you scale the lattice length.

\subsection{Results}
We informally describe the rotation-invariant Hamiltonian problem as follows. 

\begin{definition} [Rotation-invariant Hamiltonian problem (Informal)]
    Consider the Hamiltonian where a single two-body term is applied to each neighboring pair of qudits on an $r$-dimensional lattice of side length $n$. Is the ground-state energy below $a$ or greater than $b$?
\end{definition}

Our main result is that the rotation-invariant Hamiltonian problem is $\qmaexp$-complete, where $\qmaexp$ is the same as QMA except the witness and verification circuit are allowed to be exponentially large in the input size. The reason we consider $\qmaexp$ rather than QMA is that the input size to our problem is actually very small in comparison to the size of the Hamiltonian. To see this, notice that our Hamiltonian can be completely described by 1) the two-body term, 2) the dimension $r$, and 3) the lattice length $n$. The first two of these only require a constant number of bits to specify, while $n$ requires $\log n$ bits to specify. Therefore, the Hamiltonian description length is exponentially smaller than the total number of qudits. However, we still would like to allow  an ``efficient'' algorithm to run in polynomial time with respect to the number of qudits, which in turn is exponential in the input size. To accommodate this technicality, we must prove quantum hardness even when the quantum computer is allowed an exponential amount of computation time.

To prove $\qmaexp$-completeness, we use the standard method of reducing an instance $x$ of an arbitrary $\qmaexp$ problem to an instance $R(x)$ of the rotation-invariant Hamiltonian problem. It turns out that in this reduction only $n$ depends on the original problem instance $x$, so we take everything else (such as the two-body term and the lattice dimension) to be parameters rather than inputs to the problem. This is described in the more technical definition of the problem in Section \ref{sec:notation}. This differs from the standard QMA-completeness result, where the Hamiltonian terms themselves are given as input. We argue that this is a more natural setting since often one is studying a particular Hamiltonian, and so it is more suitable to consider the hardness of a given Hamiltonian for increasingly large system sizes. A desirable feature of our reduction is that the Hamiltonian we construct has no dependence on the system size or lattice dimension.

The $\qmaexp$-completeness of this problem has a number of interesting implications. First, it suggests that not even a quantum computer can find the ground-state energy of certain rotation-invariant Hamiltonians. Since nature can be viewed as a quantum computer this means that the system itself cannot find its own ground state either, suggesting the emergence of spin glass behavior at low temperatures. Next, our result implies (assuming $\qcmaexp \neq \qmaexp$) that the ground state of these Hamiltonians cannot have an efficient classical description and thus cannot be well approximated by a product state. In fact, our result directly implies a lower bound of $\Omega(n^{-r} r^{-1})$ for how close the average ground-state energy per term can be approximated by a product state. This complements the result in \cite{brandaoharrow} by providing a corresponding lower bound for the product-state approximation error.

\subsection{Techniques}

Our main approach is to carve out one-dimensional chains within our higher dimensional lattice so that we can apply \cite{v009a002}'s 1D construction to these chains. To do this, we take inspiration from the closely related classical problem of tiling. In this problem, imagine fitting together square tiles to cover an entire floor, where we are given a penalty for placing certain color tiles next to each other. The task is now to design a set of tiling rules where the least penalized configuration is a pattern of stripes. If such a set of tiling rules exists, we can use a portion of each qudit's Hilbert space to represent a tile and encode the tiling rules in our two-body Hamiltonian term. This enforces that the tiling of the ground state will have this striped pattern. Our construction then proceeds by applying the 1D construction onto neighboring qudits with same-colored tiles, which are now effectively spin chains.

Unfortunately, such a set of tiling rules does not actually exist, and so we will have to modify this classical technique to incorporate some quantum phenomena. To see this, consider the 3D case with periodic boundary conditions. No matter what set of rules are given, the optimal tiling will always take the following form. Start by tiling the first column of the first 2D slice with the optimal 1D configuration.  Then, for each subsequent column in this slice, tile it by offsetting this sequence by exactly one. Finally for each subsequent 2D slice, tile by shifting the entire previous 2D configuration by one. With some thought, one can convince themselves that this is the correct tiling. The issue is that this configuration cuts the 3D lattice into diagonal planes which is not the desired 1D structure. Additionally, this argument also shows that these rotation-invariant tiling problems are in P whereas their translation-invariant (but not rotation-invariant) counterparts are NEXP-complete for dimension 2 and higher \cite{v009a002}. This further shows how the additional symmetry constraints can potentially simplify the complexity.

It is possible to still achieve our tiling goal by combining it with some purely quantum effects. In particular, we introduce a new technique that uses the monogamy of entanglement to enforce an effective 1D geometry. This is performed by first appending two qubits to the Hilbert space of each particle. The key idea is to enforce that same-colored neighbors share an EPR pair among their qubits. Since each site only has two qubits, it can only share an EPR pair with two neighbors by the monogamy of entanglement. Thus, it can only have two same-colored neighbors. This is already close to our goal, since now our lattice must be colored by disjoint same-colored loops. It remains to make sure that these loops do not have any turns but instead cut straight across the lattice. This can be handled by imposing some further classical tiling constraints, which we discuss in more detail in Section \ref{sec:construction}. We hope that this EPR pair technique can also find use in other Hamiltonian complexity problems that benefit from embedding lower dimensional geometries into higher ones.

One last technicality to resolve is that \cite{v009a002}'s 1D Hamiltonian is frustrated and has an energy of at least $1/2$. This results in an overall energy of $n^{r-1}/2$ when this Hamiltonian is embedded into $n^{r-1}$ 1D lines in the lattice. In order to balance this energy penalty with the rest of the Hamiltonian terms, it is necessary to normalize this contribution with a coefficient that depends on $n$. However, such a system size dependence is unnatural and is preferably avoided. To fix this, we first embed a 2D translation-invariant Hamiltonian into the lattice by encoding stripes in two different directions as opposed to just one. In this case, the same result can be achieved as in the 1D case except the construction can now be made to be nearly frustration-free, removing the need for a system-size dependent normalization term.

\subsection{Outline}
We begin by describing our notation and briefly state the technical version of our result in \cref{sec:notation}. Next, we define the Hamiltonian construction in \cref{sec:construction}. Finally, the proof of our main theorem is presented in \cref{sec:analysis} where it is shown that our construction satisfies both completeness and soundness.

\section{Notation and technical result} \label{sec:notation}
    Define $\Lambda_r(n) \defeq \mathbb{Z}^r/n\mathbb{Z}^r$ to be a periodic lattice. In other words, the lattice is $r$-dimensional, where each dimension has length $n$ and each site is denoted by integer coordinates. For a lattice point $u \in \Lambda_r(n)$ we denote the $i$th coordinate as $u_i$. To define distance between points in the lattice while respecting the periodic boundary conditions, we use the Lee metric, which is defined as follows:
    \begin{definition}[Lee metric]
    Let $x,y \in \Lambda_r(n)$.
    \begin{align*}
        d(x,y) = \sum_{i=1}^r \min (|x_i - y_i|,n-|x_i - y_i|)  
    \end{align*}
    
    \end{definition}
    
    The set of nearest-neighbor pairs is defined as $E_{\Lambda_r(n)} = \{\{x,y\}: x,y \in \Lambda_r(n), d(x,y)=1\}$. If $h$ is a 2-local Hamiltonian term and $u,v \in \Lambda_r(n)$ then $h^{u,v}$ denotes the Hamiltonian term $h$ applied to sites $u$ and $v$. For an operator $H$, we denote the lowest eigenvalue of $H$ by $E_0(H)$.

    \begin{definition}[$\qmaexp$]
        A language $L$ is in $\qmaexp$ if there exists a quantum verifier $V$ such that on input $x$, $V$ has runtime $O(2^{\abs{x}^k})$ for some $k$. In addition, if $x \in L$ then there exists an $O(2^{\abs{x}^k})$-qubit state $\ket{\psi}$ such that $V(x,\ket{\psi})$ accepts with probability at least $2/3$. If $x \not\in L$, then $V(x, \ket{\psi})$ accepts with probability at most $1/3$.
    \end{definition}

With this notation in hand, the formal definition of the rotation-invariant Hamiltonian problem can be stated as follows:
\begin{definition}
\label{def:rdimrih}
r-DIM-RIH (Rotationally-Invariant Hamiltonian) \\
\textbf{Problem Parameter:} The geometric dimension of the lattice $r$. A permutation-invariant two-qudit Hermitian operator $h$. Two polynomials $p$ and $q$.\\
\textbf{Input:} Integer $n$ specified in binary.\\
\textbf{Promise: } Let $N=|\Lambda_r(n)|=n^r$. Consider the Hamiltonian $H = \sum_{\{u,v\} \in E_{\Lambda_r(n)}} h^{u,v}$. The ground state energy of $H$ is either at most $p(n)$ or at least $p(n) + 1/q(n)$.\\
\textbf{Output:} Determine whether the ground-state energy of $H$ is at most $p(n)$ or at least $p(n) + 1/q(n)$.
\end{definition}

In particular, notice that $h$ does not depend on the system size $n$ or the lattice dimension $r$ in this definition. Our main result is the following theorem.

\begin{theorem}
    r-DIM-RIH is $\qmaexp$-complete for $q(n) = 1$.
\end{theorem}
\section{Hamiltonian construction}
\label{sec:construction}
Our strategy will be first to embed a two-dimensional translation-invariant Hamiltonian without reflection symmetry into our \textit{rotationally}-invariant Hamiltonian. In the case where our lattice has dimension $r$, we will break up the lattice into $n^{r-2}$ 2D slices and embed the Hamiltonian into each of these slices. Then we can utilize the extra parameters of the 2D translation-invariant Hamiltonian to embed a one-dimensional hard Hamiltonian that is nearly frustration-free.

To accomplish this, it is crucial to have a mechanism to break the rotation symmetry by selecting a direction. We will do this by embedding two sets of directed stripes to indicate the two directions of the 2D grid.

\subsection{Embedding directed stripes}

In our construction, we will attach the following Hilbert spaces to each site in the lattice:
\begin{align*}
\mcH_{T_1} \otimes \mcH_{T_2}
\end{align*}
Our Hamiltonian will act diagonally in these subspaces and, therefore, it will only enforce classical constraints with respect to a given set of basis states, which we denote by the sets $T_1$ and $T_2$, respectively. We refer to these basis states as ``tiles" that we can assign to each site.

We define $T_1 := \{\text{red}, \text{yellow},\text{blue}\}$ and $T_2 := \{0,1,2\}$, which associate a color and number with each tile, respectively. It can be enforced that two tiles $t_1$ and $t_2$ cannot be placed next to each other by including the term $\ketbra{t_1,t_2}{t_1,t_2}$ in the Hamiltonian. In this way, we incorporate the following rules for which tiles are allowed to be placed next to each other:

\begin{enumerate}
    \item If two neighboring tiles have the same color then they must have different numbers.
    \item If two neighboring tiles have different colors then they must have the same number.
\end{enumerate}

To write down the Hamiltonian terms associated with these rules more explicitly, let $V$ represent the set of illegal neighboring tiles. Then we include the following Hamiltonian term:

\begin{align*}
    h_{\text{tile}} = 8 \sum_{(s,t) \in V} \ketbra{s,t}{s,t}
\end{align*}

The energy cost of $8$ is carefully chosen to balance out other competing terms introduced later.

\subsubsection{EPR projections}
Next, we would like to enforce that the qubits of same-colored neighbors form EPR pairs with each other. We can do this by attaching the following two additional Hilbert spaces to each site:
\begin{align*}
    \mcH_{\sigma_1} \otimes \mcH_{\sigma_2}
\end{align*}

Since each site only has two qubits it can only form two EPR pairs and therefore can only have two same-colored neighbors. Thus, this accomplishes our goal of forcing the same-colored tiles to form one-dimensional chains.

Given two same-colored neighbors $u$ and $v$, it remains to determine which of their qubits must form EPR pairs. To do this, we first define directed edges between each basis state of $T_2$ such that $0 \rightarrow 1$, $1 \rightarrow 2$ and $2 \rightarrow 0$. Due to the constraints in the previous section, $u$ and $v$ must both have different numbers. Without loss of generality, if the directed edge between these numbers points towards $v$'s tiles then we enforce that $u$'s $\sigma_2$ qubit forms an EPR pair with $v$'s $\sigma_1$ qubit. We denote this EPR pair state as $\ket{\Phi^+}_{u_{\sigma_2}v_{\sigma_1}} = \frac{1}{\sqrt{2}} (\ket{00} + \ket{11})_{u_{\sigma_2}v_{\sigma_1}}$. More formally we add the following term acting on sites $u$ and $v$: 
\begin{align*}
    A^{u,v} = 16 \sum_{i \in T_2} \ketbra{i,i+1 \bmod 3}{i,i+1 \bmod 3} \otimes (\frac{\eye - \ketbra{\Phi^+}}{2})_{u_{\sigma_2}v_{\sigma_1}}
\end{align*}
Our convention throughout is to separate sites of the lattice by commas within the braket notation and to separate subspaces within each site by tensor product symbols. Notice that if $u$ and $v$ have different numbers then this implies they also have the same color and so this does not need to be additionally conditioned on in this Hamiltonian term. In order to preserve rotation invariance, we add this term for both orderings of the particles $u$ and $v$: 
\begin{align}
    h_{\text{EPR}} =A^{u,v}+A^{v,u}
\end{align}

In addition to enforcing that contiguous regions of same-colored sites form 1D chains, these EPR constraints also require that each same-colored chain is numbered as the periodic sequence: $0,1,2,0,1,2,\dots$ either in the forwards or backwards direction (see \cref{fig:eprchain}). This is because there can never be a site where the directed edges incident to it are both pointing away or both pointing towards it. In either case, this requires one of that site's qubits to be in two different EPR pairs, which is not allowed by the monogamy of entanglement. This scenario is depicted in \cref{fig:eprupdown}. This sequential numbering, is very helpful because it defines a direction to each chain. Notice that such a numbering is only possible when $n$ is a multiple of $3$, so we will later incorporate this restriction into our hardness reduction.
\begin{figure}
    \centering
    \includegraphics[width=0.7\columnwidth]{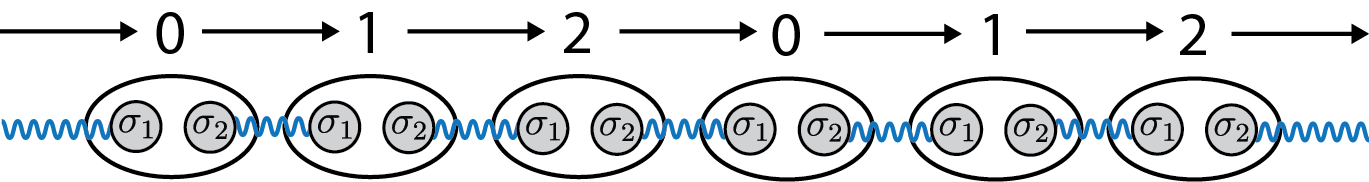}
    \caption{When a chain of sites is numbered sequentially around the cycle $\mathbb Z_3$, each qubit is matched with exactly one other qubit to form an EPR pair, and so the EPR constraint can easily be satisfied.}
    \label{fig:eprchain}
\end{figure}

\begin{figure}
    \centering
    \includegraphics[width=0.2\columnwidth]{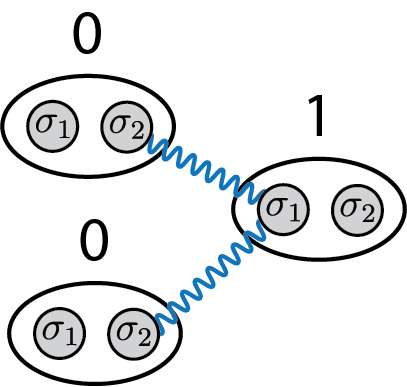}
    \caption{When there are three consecutive same-color neighbors that are not numbered monotonically around $\mathbb Z_3$ (for instance $0,1,0$) then two different qubits are matched with the same qubit to form an EPR pair. Due to the monogamy of entanglement this constraint cannot be satisfied and incurs an energy penalty.}
    \label{fig:eprupdown}
\end{figure}

\subsubsection{1D chain boundary conditions}
Given the current Hamiltonian terms, the same-colored sites can form chains with either open or periodic boundary conditions (lines or loops). It will be convenient later that the boundary conditions are fixed and so we add another term to enforce periodic boundary conditions. In particular, we define the following term acting on sites $u$ and $v$:
\begin{align*}
    h_{\text{loop}} = 2\sum_{c,d \in T_1 \text{ } | c \neq d}\ketbra{c,d}{c,d},
\end{align*}
which penalizes neighboring sites with different colors. This adds a penalty of $2r-2$ for every site in the middle of the chain. This is because each site has $2r$ neighbors, where all but exactly two are colored differently. This results in a penalty of $2(2r-2)/2$ where we have divided by two to fix double-counting. Using similar reasoning, a penalty of $2r-1$ is incurred for every site at the endpoint of an open chain. Assuming that the classical tiling and EPR constraints are satisfied, this term is optimized when all 1D chains form loops so that each site always neighbors two other sites of the same color. The scaling of $2$ is chosen so that the energy savings of coloring neighboring sites the same will never outweigh the energy penalty of violating the EPR constraints. This tradeoff will be worked out in detail in Section \ref{sec:soundness}.

\subsection{Adding the second dimension}
Now that we have embedded stripes in one direction of the lattice, we must repeat this process to embed stripes in another direction. To do this, we can simply make a copy of each site's Hilbert space and apply the same Hamiltonian terms to the copy. It remains to ensure that both copies do not have stripes oriented in the same direction. To do this, it is sufficient to simply disallow two neighboring tiles from having matching colors on both copies. In other words, we add the following term where we let $\mcH_{T_1}$ denote the $T_1$ subspace of the first copy and $\mcH'_{T_1}$ denote that of the second copy.

\begin{align*}
    h_{\text{copy}} = (\sum_{c \in T_1}\ketbra{c,c}{c,c})_{\mcH_{T_1}} \otimes (\sum_{d \in T_1}\ketbra{d,d}{d,d})_{\mcH'_{T_1}}
\end{align*}

\begin{figure}[t!]
    \centering
    \begin{subfigure}[t]{0.225\textwidth}
        \centering
        \includegraphics[width=\textwidth]{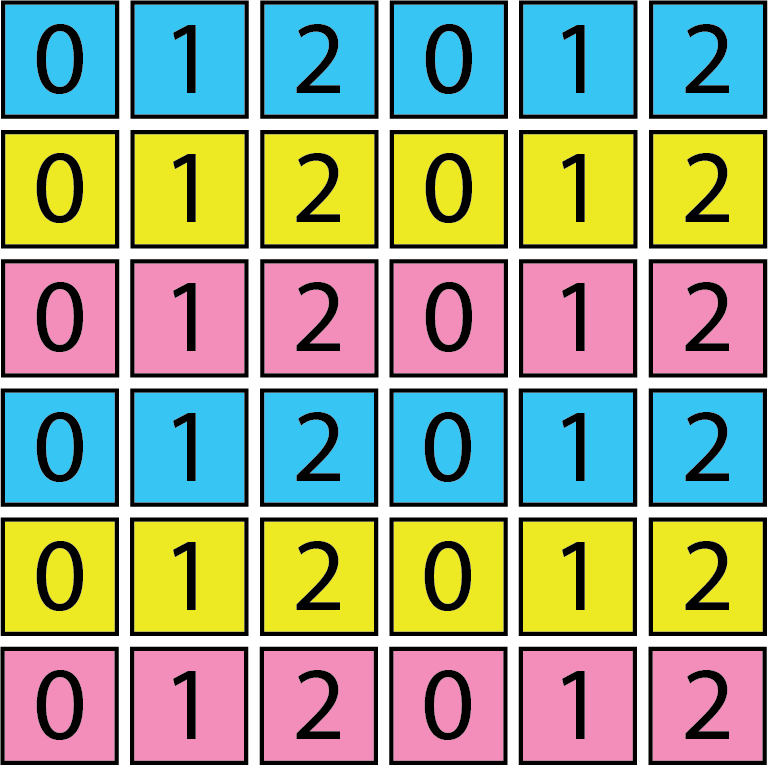}
        \caption{First copy of tiles}
    \end{subfigure}
    \hspace{0.2\textwidth}
    \begin{subfigure}[t]{0.225\textwidth}
        \centering
        \includegraphics[width=\textwidth]{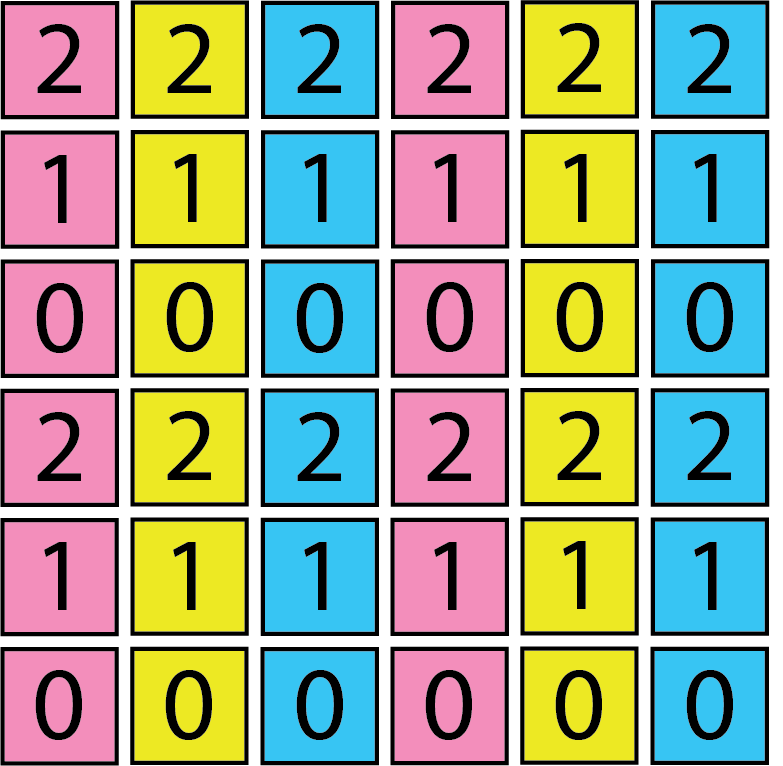}
        \caption{Second copy of tiles}
    \end{subfigure}
    \caption{An example of how a 2D lattice can be tiled to optimize the Hamiltonian terms. Specifically, each copy forms a striped pattern where each stripe is numbered in a cyclic sequence. In addition, the two copies of tiles must have stripes pointing in different directions. Finally, the rows/columns that do not hold stripes must be numbered with the same number. Now a translation-invariant Hamiltonian can be simulated by using these tilings as a guideline for which sites are above, below, left, or right of each other. As our convention, we take the first copy to denote the horizontal direction and the second copy to denote the vertical direction.}
    \label{fig:2dcolor}
\end{figure}

\subsection{Embedding the translation-invariant Hamiltonian}

An arbitrary two-dimensional translation-invariant Hamiltonian $H_{\text{TI}}$ can now be embedded into our rotation-invariant Hamiltonian by using the directed stripes as guidelines. Our convention will be to let the first copy of each site represent the horizontal stripes and the second copy represent the vertical stripes. These tiling patterns are depicted in \cref{fig:2dcolor}. Now, the horizontal Hamiltonian term is applied only when the first copy tiles have the same color (i.e., different numbers). In addition, the orientation of which site is on the left and which is on the right will be decided by the directed edge in between the two tile's numbers where each arrow points from left to right. The vertical Hamiltonian term is applied similarly with respect to the second copy of each site's tiles.

To make this more concrete, we first attach to each site of our lattice the Hilbert space of a site in $H_{\text{TI}}$ which we denote by $\mcH_{\text{2D}}$. Denote the horizontal 2-body term of $H_{\text{TI}}$ by $h_{\text{TI}}$. To incorporate this into our construction we add the following term:
\begin{align}
h_{\text{RI}} = \sum_{i \in T_2} \ketbra{i,i+1 \bmod 3} \otimes h_{\text{TI}} \\
+  \sum_{i \in T_2} \ketbra{i,i-1 \bmod 3} \otimes Sh_{\text{TI}}S
\end{align}
where $S$ is the swap operator on the two sites and the term acts on the first copy of $\mcH_{T_2}$. Note that while $h_{\text{TI}}$ does not have reflection symmetry, $h_{\text{RI}}$ does. The vertical term is implemented in the same way but acts on the second copy. We denote the vertical term by $v_{\text{RI}}$.  

It remains to now describe the 2D translation-invariant Hamiltonian that encodes the computational hardness. We will combine techniques from \cite{v009a002} in order to encode the desired ground-state energy in the yes and no cases. The details are deferred to \cref{app:2d} but the key result is outlined below.

\begin{theorem}
\label{thm:GI}
      Let $L$ be a $\qmaexp$-complete language. There exists an efficiently computable function $f: \{0,1\}^* \rightarrow \mathbb{Z}$ and 2-local positive semidefinite Hamiltonian terms $h_{\text{TI}}$ and $v_{\text{TI}}$ with the following properties:
    \begin{enumerate}
        \item $f(x)$ is a multiple of $3$ and $f(x)/3$ is prime. Furthermore, $\log f(x)=O(\poly(|x|))$ and $f$ is computable in $O(\poly(|x|))$ time.
        \item Let $f(x) = n \geq n_0$ for some constant $n_0$ and a given problem instance $x \in \{0,1\}^*$. For an $n \times n$ 2D lattice with periodic boundary conditions, let $E_h$ be the set of ordered pairs of horizontal neighbors and let $E_v$ be the set of ordered pairs of vertical neighbors. Consider the Hamiltonian $H_{\text{TI}}=\sum_{(u,w) \in E_h} h_{\text{TI}}^{u,w} + \sum_{(x,y) \in E_v} v_{\text{TI}}^{x,y}$. 
        \begin{enumerate}
            \item If $x \in L$ then $E_0(H_{\text{TI}}) \leq O(n^{-k})$ for an arbitrarily large constant $k$.
        \item If $x \not\in L$ then $E_0(H_{\text{TI}}) \geq \Omega(1/n^3)$
        \end{enumerate}
    \end{enumerate}
\end{theorem}

Using the prime symbol to denote terms acting on the second copy of a site's Hilbert space, we can now write the entire 2-body Hamiltonian term of our construction as follows:
\begin{align}
\label{eq:h}
    h = h_{\text{tile}}+h_{\text{EPR}}+h_{\text{loop}} + h'_{\text{tile}}+h'_{\text{EPR}}+h'_{\text{loop}} + h_{\text{copy}}+h_{\text{RI}}+v_{\text{RI}}
\end{align}

We can now define the reduction from any $\qmaexp$ problem instance to an r-DIM-RIH problem instance.
\begin{definition}[Reduction from $\qmaexp$ to r-DIM-RIH]
    Let $L$ be a language in $\qmaexp$ and let $x \in \{0,1\}^*$ be a problem instance. Define the function $R(L,x,r) = \sum_{\{u,v\} \in E_{\Lambda_r(f(x))}} h^{u,v}$ where $f: \{0,1\}^* \rightarrow \mathbb{Z}$ is constructed as in Theorem \ref{thm:GI} and $h$ as in Eq. \ref{eq:h}.
\end{definition}

Next, the detailed analysis of this construction is provided.

\section{Analysis}
\label{sec:analysis}

In this section we prove the main theorem that the rotation-invariant Hamiltonian problem is $\qmaexp$-complete.

\subsection{Completeness}
We begin by first considering the case where $x \in L$. In this case, we present a ground state that achieves energy below $p(n) = 4 n^r (r-1)+1/g(n)$ for an arbitrarily large polynomial $g(n)$.
\begin{lemma}
    Let $L$ be a language in $\qmaexp$ and let $x \in \{0,1\}^*$ be a problem instance. Let $H = R(L,x,r)$. If $x \in L$, then $E_0(H) \leq 4 n^r (r-1) +1/g(n)$ for any polynomial $g(n)$.
\end{lemma}
\begin{proof}
    Consider the following tiling, which generalizes the tiling of the 2D lattice depicted in \cref{fig:2dtiles}. We first construct a $3$-coloring of the $(r-1)$-dimensional lattice. This always exists because an $(r-1)$-dimensional lattice can be decomposed as a Cartesian product of cycle graphs. Since each cycle graph is $3$-colorable their Cartesian product is also $3$-colorable by a result by Sabidussi \cite{Sabidussi_1957}. To color a site in the full $r$-dimensional lattice we simply drop the $1$st coordinate and assign the color from the $3$-coloring of the remaining $r-1$ coordinates. This enforces that any two neighboring sites that have the same $1$st coordinate are colored by different colors and any two neighbors that only differ in the $1$st coordinate are tiled by the same colors. This has the effect of coloring 1D chains of sites that travel in a straight line along the $1$st coordinate dimension.

    For this coloring, it is easy to satisfy the numbering constraints. This can be accomplished by tiling all particles $u$ with the number $u_1 \text{ mod } 3$. This simultaneously tiles all 1D chains with the ordering $0,1,2,0,1,2,\dots$ which eventually wraps around since $n$ is a multiple of $3$. In addition, all neighboring tiles with different color tiles are tiled with the same number since this only occurs when the two particles have the same $1$st coordinate.

    The EPR constraint can easily be satisfied since the particles have been tiled as disjoint 1D chains with the appropriate numbering. In addition, all chains are loops and so the constraint on having periodic boundary conditions is also satisfied. This ensures that only a penalty of $ n^r(2r-2)=2 n^r (r-1)$ is introduced by the $h_{\text{loop}}$ term. We can repeat this for the second copy of tiles but now directing the 1D chains along the second coordinate. This introduces another penalty of $2 n^r (r-1)$.

    With this choice of tiles the lattice has effectively been broken up into $n^{r-2}$ 2D slices where we can now apply the 2D translation-invariant Hamiltonian construction. By \cref{thm:GI}, since $x \in L$ we have that each $2D$ slice contributes an energy of at most $O(n^{-k})$. The total energy is thus $O(n^{r-2 - k}) = O(n^{-k'})$ where we have chosen $k = r-2+k'$. This results in a final energy upper bound of $4 n^r (r-1) + O(n^{-k'})$.
\end{proof}
\subsection{Soundness} 
\label{sec:soundness}
In this section we will prove the following lemma:
\begin{lemma} \label{lemma:soundness}
    Let $L$ be a language in $\qmaexp$ and let $x \in \{0,1\}^*$ be a problem instance. Let $H = R(L,x,r)$. If $x \not \in L$, then $E_0(H) \geq 4n^r(r-1) +1$.
\end{lemma}
Our general strategy will be to first lower bound the energy of any state with a given classical tiling. We call these ``tile states" and define them as follows:
\begin{definition}[Tile state]
    A \textbf{tile state} is a state $\ket{\psi_{c,c'}} = \ket{c} \otimes \ket{c'} \otimes \ket{\phi}$ where $c,c' \in T_1^{N} \times T_2^{N}$ and $\ket{\phi} \in (\mcH_{\sigma_1} \otimes \mcH_{\sigma_2} \otimes \mcH'_{\sigma_1} \otimes \mcH'_{\sigma_2} \otimes \mcH_{\text{2D}})^{\otimes N}$. The notation $\ket{\psi_c}$ is used whenever only one of the tilings is relevant.
\end{definition}
Lower bounding the energy of an arbitrary state follows straightforwardly, since each of the Hamiltonian terms is diagonal with respect to the tile Hilbert spaces.

We start by first establishing the fact that a qubit cannot be in two EPR pairs at once.
\begin{fact}
    \label{fact:eprlb}
    Consider a Hamiltonian on three qubits $u$, $v$ and $w$ defined by $(\frac{\eye - \ketbra{\Phi^+}}{2})_{u,v} + (\frac{\eye - \ketbra{\Phi^+}}{2})_{v,w}$. By direct computation this has a minimum eigenvalue of $1/4$.
\end{fact}
We next focus on one copy of the tilings at a time and define the following notation.

\begin{definition}
    For a given site $u$, let $n_u$ be the number of $u$'s neighbors that are tiled with the same color as $u$. 
\end{definition}
With these in hand, we can now lower bound the energy of any tile state in terms of the number of same-colored neighbors of each site.
\begin{claim}
    \begin{align}
    \label{eq:h1lb}
    \bra{\psi_c}\sum_{\{u,v\} \in E_{\Lambda_r(n)}} (h_{\text{tile}}+h_{\text{EPR}}+h_{\text{loop}})^{u,v} \ket{\psi_c} \geq 
    2n^r r - \sum_{u \in \Lambda_r(n)}n_u +4\sum_{u \in \Lambda_r(n)} \lfloor n_u/3 \rfloor
    \end{align}
\end{claim}
\begin{proof}
    Recall that $h_{\text{loop}}$ gives an energy penalty of $2$ for neighboring particles that are tiled by opposite colors. Every particle has $2r$ neighbors and so if every particle is tiled with a different color than all of its neighbors then the total penalty due to this Hamiltonian term would be $2\frac{n^r(2r)}{2} = 2n^r r$. The total number of neighboring pairs with the same color tiling is $\sum_{u \in \Lambda_r(n)} \frac{n_u}{2}$. Since each of these pairs saves an energy of $2$, the total energy with respect to $h_{\text{loop}}$ applied to each edge is $2n^r r - 2\sum_{u \in \Lambda_r(n)} \frac{n_u}{2}=2n^r r - \sum_{u \in \Lambda_r(n)} n_u$. 
    
    Now for every particle $u$ we can group its same-colored neighbors into groups of three until a full group of three cannot be formed. There will be $\lfloor n_u/3 \rfloor$ such groups. For a given group of three neighbors, label the particles $v,w,z$. If $h_{\text{tile}}$ applied to edge $\{u,v\}$, $\{u,w\}$, or $\{u,z\}$ is violated then this incurs a penalty of $4$ per particle involved ($8$ in total for the edge). 
    
    If none are violated then this means at least two of the three particles, $v,w,z$, must be tiled with the same number. This is because $h_{\text{tile}}$ enforces that same-colored neighbors of $u$ are tiled with a different number than $u$ and there are only two such numbers. Without loss of generality assume that $v$ and $w$ are tiled with the same number and it is exactly one higher (mod $3$) than $u$'s number. Then we have the following bound
    \begin{align}
        E_0(h^{u,v}_{\text{EPR}}+h^{u,w}_{\text{EPR}}) &\geq  E_0(A^{u,v}+A^{u,w})\\
        &\geq 16E_0((\frac{\eye - \ketbra{\Phi^+}}{2})_{u_{\sigma_2}v_{\sigma_1}} + (\frac{\eye - \ketbra{\Phi^+}}{2})_{u_{\sigma_2}w_{\sigma_1}})\\
        &\geq 4 &\text{By Fact \ref{fact:eprlb} }
    \end{align}
    Repeating this argument for each site $u$ in the lattice does not double count energy penalties since we only used the $A^{u,v}$ term of $h_{\text{EPR}}$, and so when considering $v$ we would use the $A^{v,u}$ term instead. Therefore, either $h_{\text{tile}}$ or $h_{\text{EPR}}$ must be violated and either way a penalty of at least $4$ is added to the energy. This argument can be repeated for each group of three same-colored neighbors and so this contributes at least $4\sum_{u \in \Lambda_r(n)} \lfloor n_u/3 \rfloor$ to the energy.
\end{proof}
It follows immediately from \cref{eq:h1lb} that each particle must have exactly two same colored-neighbors. Otherwise this induces an energy penalty of at least one which is enough to imply the bound in \cref{lemma:soundness}. This is helpful since we have now shown that the tiling pattern forms loops. It still remains to show that these loops are straight. Before moving on to this, we first make the above statements rigorous as follows:
\begin{definition}
    A classical tiling $c \in T_1^{N} \times T_2^{N}$ is \textbf{looped} if $n_u = 2$ $\forall u \in \Lambda_r(n)$
\end{definition}
\begin{claim}
    \label{claim:prop1}
    If $c$ or $c'$ is not looped, then $ \bra{\psi_{c,c'}} H \ket{\psi_{c,c'}} \geq 4 n^r(r-1) + 1$.
\end{claim}
\begin{proof}
    Without loss of generality assume that $c$ is not looped and that $c'$ is any tiling. The inequality in Eq. \ref{eq:h1lb} is minimized when $n_u=2$ for all $u \in \Lambda_r(n)$. This is because $-n_u + 4\lfloor n_u/3 \rfloor = -2$ for $n_u = 2$ and $-n_u + 4 \lfloor n_u/3 \rfloor \geq -1$ for $n_u \neq 2$ where $n_u$ is a nonnegative integer. The right-hand side of Eq. \ref{eq:h1lb} is equal to $2n^r r - 2n^r = 2n^r(r-1)$ when $\forall u$ $n_u = 2$. Therefore, when $n_u\neq2$ for some $u \in \Lambda_r(n)$, $ \bra{\psi_c}\sum_{\{u,v\} \in E_{\Lambda_r(n)}} (h_{\text{tile}}+h_{\text{EPR}}+h_{\text{loop}})^{u,v}  \ket{\psi_c} \geq 2 n^r (r-1)+1$. The terms $h'_{\text{tile}},h'_{\text{EPR}},$ and $h'_{\text{loop}}$ add a penalty of at least $2n^r(r-1)$ as we have just argued. Finally, each term in $H$ is positive semidefinite and so the energy with respect to the entire Hamiltonian is lower bounded by the energy with respect to a subset of the terms. The claim then follows.
    
\end{proof}

We now must show that these loops are in fact straight and do not contain any turns.
\begin{definition}
    Let $g: \Lambda_r(n)^2 \rightarrow \mathbb{N}$ be a function that outputs the number of coordinates that differ between two lattice sites.
\end{definition}
\begin{definition}
    A classical tiling $c \in T_1^{N} \times T_2^{N}$ has a \textbf{turn} if there exists three particles $u$, $v$, and $w$ where the following is true: $u$ and $w$ are neighbors of $v$ (i.e. $d(u,v) = d(v,w) = 1$), $u$, $v$, and $w$ are all tiled by the same color, and $g(u,w) =2$.
\end{definition}
We can now penalize turns by using the tiling rule that different color neighbors must have the same number. This is because if there is a turn then there exists some neighboring site outside of the loop that borders two sites right where the turn occurs. This causes a penalty because both sites in the loop must have the same number as the site outside of the loop, which contradicts the sequential numbering of the loop. This is argued in more detail below.
\begin{claim}
    \label{claim:noprop1prop2}
    If $c$ or $c'$ is looped but has a turn, then
    \begin{align}
        \bra{\psi_{c,c'}} H \ket{\psi_{c,c'}} \geq 4 n^r (r-1)+4
    \end{align}
\end{claim}
\begin{proof}
    Once again without loss of generality, assume that $c$ is looped and has a turn and that $c'$ is any tiling. If $c$ has a turn then there exists $u$, $v$, and $w$ that are all tiled by the same color and $g(u,w) = 2$. Without loss of generality, let $u_1=0$, $u_2=0$, $v_1=0$, $v_2=1$, $w_1=1$, $w_2=1$. This simplification is for clarity but note that the following argument works for any coordinates such that $d(u,v) = d(v,w)  = 1$ and $g(u,w) = 2$. Consider a fourth site $z$ at coordinates $z_1=1$, $z_2=0$. Note that $d(z,u)  = d(z,w) =1$ and so it neighbors $u$ and $w$. This situation is depicted in \cref{fig:turn}. If $z$ is tiled by the same color as $u$, $v$, and $w$ then it would form a loop of length four. This violates $h_\text{tile}$ since the loop is not a multiple of three and so the numbers can not sequentially wrap around $\mathbb{Z}_{3}$ (see \cref{fig:turna}). If $z$ is tiled with a different color than $u$, $v$, and $w$ then it must be tiled with the same number as $u$ and $w$ since $h_{\text{tile}}$ enforces that different colored neighbors must have the same number. However, as we have argued previously, a site cannot have two same-colored neighbors that are tiled with the same number as each other since this causes the EPR constraint to be violated (see \cref{fig:turnb}). In either case, there is an energy penalty of at least $4$. In addition, the $h_{\text{loop}}$ and $h'_{\text{loop}}$ terms together add a penalty of $4 n^r (r-1)$. Finally, noting the positive semidefiniteness of each Hamiltonian term concludes the proof.
\end{proof}

\begin{figure}[t!]
    \centering
    \begin{subfigure}[t]{0.15\textwidth}
        \centering
        \includegraphics[width=\textwidth]{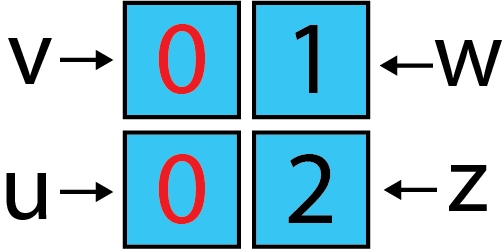}
        \caption{}
        \label{fig:turna}
    \end{subfigure}
    \hspace{0.1\textwidth}
    \begin{subfigure}[t]{0.15\textwidth}
        \centering
        \includegraphics[width=\textwidth]{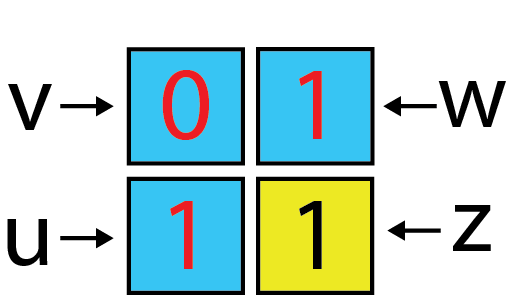}
        \caption{}
        \label{fig:turnb}
    \end{subfigure}
    \caption{Two examples of how energy penalties can arise when there is a turn in the loop. Here, the loop consists partially of $u$, $v$, and $w$ which contains a turn since $u$ and $w$ differ in more than one coordinate. In (a), $z$ is colored the same as the rest but this results in a penalty since a loop of size $4$ cannot be numbered cyclically around $\mathbb Z_3$. This results in the illegal configuration of $u$ and $v$ having the same color and the same number. In (b), $z$ is colored differently but this leads to part of the loop being numbered as $1,0,1$, which is also illegal as depicted in \cref{fig:eprupdown}.}
    \label{fig:turn}
\end{figure}

So far we have given a sufficient lower bound for any tile states that do not consist of only straight loops. It will be helpful to make a few more observations on the structure of these 1D loops. 
\begin{definition}
     A classical tiling $c \in T_1^{N} \times T_2^{N}$ is \textbf{uniformly directed} if it is looped without turns and each 1D loop is oriented along the same dimension.
\end{definition}

\begin{claim}
    If $c$ or $c'$ is looped without turns but not uniformly directed then 
    \begin{align}
        \bra{\psi_{c,c'}} H \ket{\psi_{c,c'}} \geq 4 n^r (r-1)+8
    \end{align}
\end{claim}
\begin{proof}
    If the tiling is not uniformly directed then the following situation will necessarily arise, which is illustrated in \cref{fig:unidir}. Without loss of generality, consider a square of sites within the lattice such that the top two sites are the same color and thus a part of the same 1D chain, but the bottom two sites have two different colors from the top and from each other. Since different color tiles must have the same numberings, then all four squares would be required to have the same number. However, since the top two have the same color, they are required to have different numbers and so it is impossible to assign numbers that do not incur any penalties.
\end{proof}

\begin{figure}
    \centering
    \includegraphics[width=0.08\columnwidth]{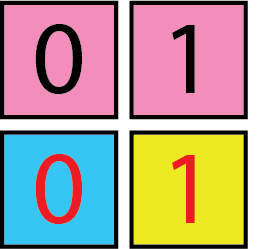}
    \caption{This configuration always arises if the tiling is not uniformly directed since otherwise each loop is always pointed in the same direction. This causes a rule violation since the blue and yellow tiles must have the same number but are forced to hold different numbers due to the red tiles.}
    \label{fig:unidir}
\end{figure}

One final property we will need is that each 1D loop is numbered consistently.

\begin{definition}
     A classical tiling $c \in T_1^{N} \times T_2^{N}$ is \textbf{numbered consistently} if it is uniformly directed towards a given dimension $d$ and every site with the same $d$th coordinate value has the same number.
\end{definition}

\begin{claim}
    If $c$ or $c'$ is uniformly directed, but not numbered consistently then 
    \begin{align}
        \bra{\psi_{c,c'}} H \ket{\psi_{c,c'}} \geq 4 n^r (r-1)+8
    \end{align}
\end{claim}
\begin{proof}
    Consider the $r-1$ dimensional sublattice of sites that each have the same $d$th coordinate value. Each pair of neighbors in this lattice must have different colors. Otherwise, the coloring would not be uniformly directed towards the $d$th dimension since this means there is some 1D loop that points in a different direction. Since each pair of numbers has different colors, they all must have the same number, i.e. the tiles must be numbered consistently. Otherwise, a penalty of at least $8$ is incurred.
\end{proof}

All that remains is to now lower bound the ground-state energy when both tilings are numbered consistently.

\begin{claim}
    \label{claim:noprop1noprop2}
    Let $c$ and $c'$ be numbered consistently. Let $L$ be a language in $\qmaexp$ and let $x \in \{0,1\}^*$ be a problem instance. Let $H = R(L,x,r)$
    If $x \not \in L$, then $\bra{\psi_{c,c'}} H \ket{\psi_{c,c'}} \geq 4 n^r (r-1) + \Omega(n^{r-5})$.
\end{claim}
\begin{proof}
    First, we must handle the case where $c$ and $c'$ both have 1D chains pointing in the same direction. This would incur a penalty of $n^rr$ from the $h_{\text{copy}}$ term alone, which would clearly imply the desired lower bound. Next, we focus on the case where they point in different directions. We let the 1D chains in the first tiling represent the horizontal rows of each 2D slice and those of the second tiling represent vertical rows. In addition, we let the order of the numberings define the left, right, up and down directions of the slices. In this way we can embed $n^{r-2}$ 2D translation-invariant Hamiltonians. Since $x \not\in L$, by \cref{thm:GI}, these terms contribute an energy penalty of $\Omega(n^{r-5})$.

\end{proof}
To complete the proof it remains to deal with the non-tile states but these can easily be handled since the Hamiltonian is diagonal in the tile Hilbert spaces.
\begin{lemma} [Restatement of \cref{lemma:soundness}]
    Let $L$ be a language in $\qmaexp$ and let $x \in \{0,1\}^*$ be a problem instance. Let $H = R(L,x,r)$. If $x \not \in L$, then $E_0(H) \geq 4n^r(r-1) + 1$.
\end{lemma}
\begin{proof}
    First we can write an arbitrary state as a superposition of tile states: $\ket{\zeta} = \sum_i \alpha_i \ket{c_i} \ket{\phi_i}$ where $c_i \in T_1^{2N} \times T_2^{2N}$. Note that $\bra{c_i} \bra{\phi_i} H \ket{c_j} \ket{\phi_j} = 0$ for $c_i \neq c_j$. This is because all terms are diagonal on the Hilbert space of the classical tiles. Therefore, we have $\bra{\zeta} H \ket{\zeta}=\sum_i \abs{\alpha_i}^2 \bra{c_i} \bra{\phi_i} H \ket{c_i} \ket{\phi_i}$. This is an affine combination of tile state energies which we have already lower bounded by $4n^r(r-1) + 1$. Thus, the energy itself is also bounded from below by $4n^r(r-1) + 1$.
\end{proof}

\section{Open boundary conditions}
So far we have only considered the case where the Hamiltonian has periodic boundary conditions. It turns out that the same construction also works for open boundary conditions. Our method of embedding directed stripes still works in this case except now instead of closed loops the stripes form spin chains with open boundary conditions. This leaves the sites at the ends of the chain with one unpaired qubit that can still form an EPR pair with another site; however, this would require introducing a turn. The energy penalty for having a turn outweighs the energy bonus of having one more same-colored neighbor and so it is optimal to leave the qubit unpaired. Thus, the same construction can once again be used to embed a 2D translation-invariant non-reflection-invariant Hamiltonian with the exception that this Hamiltonian now has open boundary conditions. To complete the result, it remains to show the equivalent of \cref{thm:GI} in the case of open boundary conditions. The proof of this statement is shown in \cref{app:open}.

\section{Conclusion}
In this work, we have resolved the complexity for rotation-invariant Hamiltonians with constant lattice dimension, but it still remains interesting to better understand the complexity at even higher lattice dimensions. For instance, we know that as $r \rightarrow \infty$ the ground state becomes a product state, but how fast does it converge? Consider the problem where the lattice length is now fixed, and the lattice dimension is given as input. Is there a small enough promise gap for which this problem is quantumly hard? Another direction to consider is to study more general permutation symmetries. In some sense, the rotation-invariant Hamiltonian problem interpolates between systems with comparatively low symmetry in the one-dimensional case to highly symmetric systems as the lattice dimension increases. It would be an interesting question to generalize this interpolation and probe whether there exists a complexity phase transition with respect to some symmetry parameter.
\section*{Acknowledgements}
JN is supported by the  National Science Foundation Graduate Research Fellowship Program under Grant No. DGE  2236417.


\bibliography{bibliography}

\appendix
\section{Constructing the 2D translation-invariant Hamiltonian}
\label{app:2d}
\subsection{Periodic boundary conditions}
\label{app:periodic}
In this section we give a proof sketch of the following theorem:
\begin{theorem}[Restatement of \cref{thm:GI}]
          Let $L$ be a $\qmaexp$-complete language. There exists an efficiently computable function $f: \{0,1\}^* \rightarrow \mathbb{Z}$ and 2-local positive semidefinite Hamiltonian terms $h_{\text{TI}}$ and $v_{\text{TI}}$ with the following properties:
    \begin{enumerate}
        \item $f(x)$ is a multiple of $3$ and $f(x)/3$ is prime. Furthermore, $\log f(x)=O(\poly(|x|))$ and $f$ is computable in $O(\poly(|x|))$ time.
        \item Let $f(x) = n \geq n_0$ for some constant $n_0$ and a given problem instance $x \in \{0,1\}^*$. For an $n \times n$ 2D lattice with periodic boundary conditions, let $E_h$ be the set of ordered pairs of horizontal neighbors and let $E_v$ be the set of ordered pairs of vertical neighbors. Consider the Hamiltonian $H_{\text{TI}}=\sum_{(u,w) \in E_h} h_{\text{TI}}^{u,w} + \sum_{(x,y) \in E_v} v_{\text{TI}}^{x,y}$. 
        \begin{enumerate}
            \item If $x \in L$ then $E_0(H_{\text{TI}}) \leq O(n^{-k})$ for an arbitrarily large constant $k$.
        \item If $x \not\in L$ then $E_0(H_{\text{TI}}) \geq \Omega(1/n^3)$
        \end{enumerate}
    \end{enumerate} 
\end{theorem}
Much of the proof will directly utilize techniques from \cite{v009a002}. When needed, a brief summary of these ideas is given but we direct the interested reader to \cite{v009a002} for the full details. The main idea of the construction is to use the 2D translation-invariant tiles to embed a 1D chain with a designated starting tile. This allows us to avoid the $1/2$ additive penalty that is required in the 1D construction when there is no designated starting tile (see section 6 ``The quantum cycle" of \cite{v009a002}). Fortunately, the construction in tables 5 and 6 of section 4.1 of \cite{v009a002} accomplish exactly this. This construction is quite complicated and so we only present the end result. The last piece to handle is that this construction requires the grid length to be prime while our tiling rules require it to be a multiple of three. This can easily be remedied by simulating each tile in the 2D construction with a $3\times3$ grid of nine tiles. We explain each of these steps in more detail below. 

We start by describing how to construct the function $f$ referenced in the theorem statement. \cite{v009a002} show that given $x$ there is a randomized algorithm $a$ running in expected time $O(\poly |x|)$ to find a prime number $p$ such that the $1/3$ most significant digits represent $x$ and $\log p = O(\poly(|x|))$. Using this result, we simply define $f$ as $f:=3\cdot a(x)$. 

Next, to construct the 2D translation-invariant Hamiltonian, we start by using the below result.
\begin{lemma}[see section 4.1 of \cite{v009a002}]
    There exists a set of translation-invariant, non-reflection-invariant tiling rules involving a constant number of total tile types on an $n \times n$ 2D grid with periodic boundary conditions such that when $n$ is prime exactly one row must contain tiles of the form $\horizontaltile$ and $\centertile$ and no other row can contain either of these tiles. In addition, exactly one site in this row must contain $\centertile$ and all other sites in this row must contain $\horizontaltile$. This configuration is depicted in \cref{fig:2dtiles}.
\end{lemma}

\begin{figure}
    \centering
    \includegraphics[width=0.2\columnwidth]{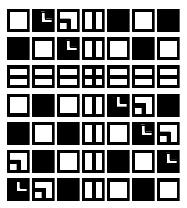}
    \caption{An example of an allowed configuration for a given set of translation-invariant tiling rules on a 2D grid defined in \cite{v009a002}. The key feature is that when the side length is a prime number, exactly one row can contain the $\horizontaltile$ and $\centertile$ tiles and one site within this row can contain the $\centertile$ tile. The above image is directly reproduced from \cite{v009a002} under CC-BY 3.0.}
    \label{fig:2dtiles}
\end{figure}

Since this construction only works for prime $n$ but our lattice length must be a multiple of three, it is necessary to replace each tile with a $3 \times 3$ grid of tiles that serve the same function as the original. First, for each tile in the original tiling, a corresponding center tile is defined as $\middletile$. Then, the following tiles and rules are added. $\toptile$ must be placed above $\middletile$. Similarly, $\bottomtile$ must be placed below $\middletile$, $\lefttile$  placed to the left of it and $\righttile$ placed to the right of it. Now to fix the corners in place we enforce that $\toplefttile$ must be placed to the left of $\toptile$, $\toprighttile$ placed to the right of $\toptile$, $\bottomlefttile$ placed to the left of $\bottomtile$ and $\bottomrighttile$ placed to the right of $\bottomtile$. In other words, the center tile must always be surrounded by the border tiles. The reciprocal of each of these rules is also included. This enforces that the border tiles must always be accompanied by the center tile in the appropriate location. This results in \cref{fig:3x3sim} being the only allowed configuration. The last thing to resolve is to ensure that these $3 \times 3$ grids are aligned with each other. To do this, we must regulate which border tiles can be placed next to those of a different $3\times3$ grid. We enforce that only $\righttile$ type tiles are allowed to the left of $\lefttile$ type tiles. We incorporate the same rule for the other sides as well: only $\lefttile$ type tiles are allowed to the right of $\righttile$ type tiles, only $\bottomtile$ type tiles are allowed above $\toptile$ type tiles and only $\toptile$ type tiles are allowed below $\bottomtile$ type tiles.

\begin{figure}
    \centering
    \includegraphics[width=0.15\columnwidth]{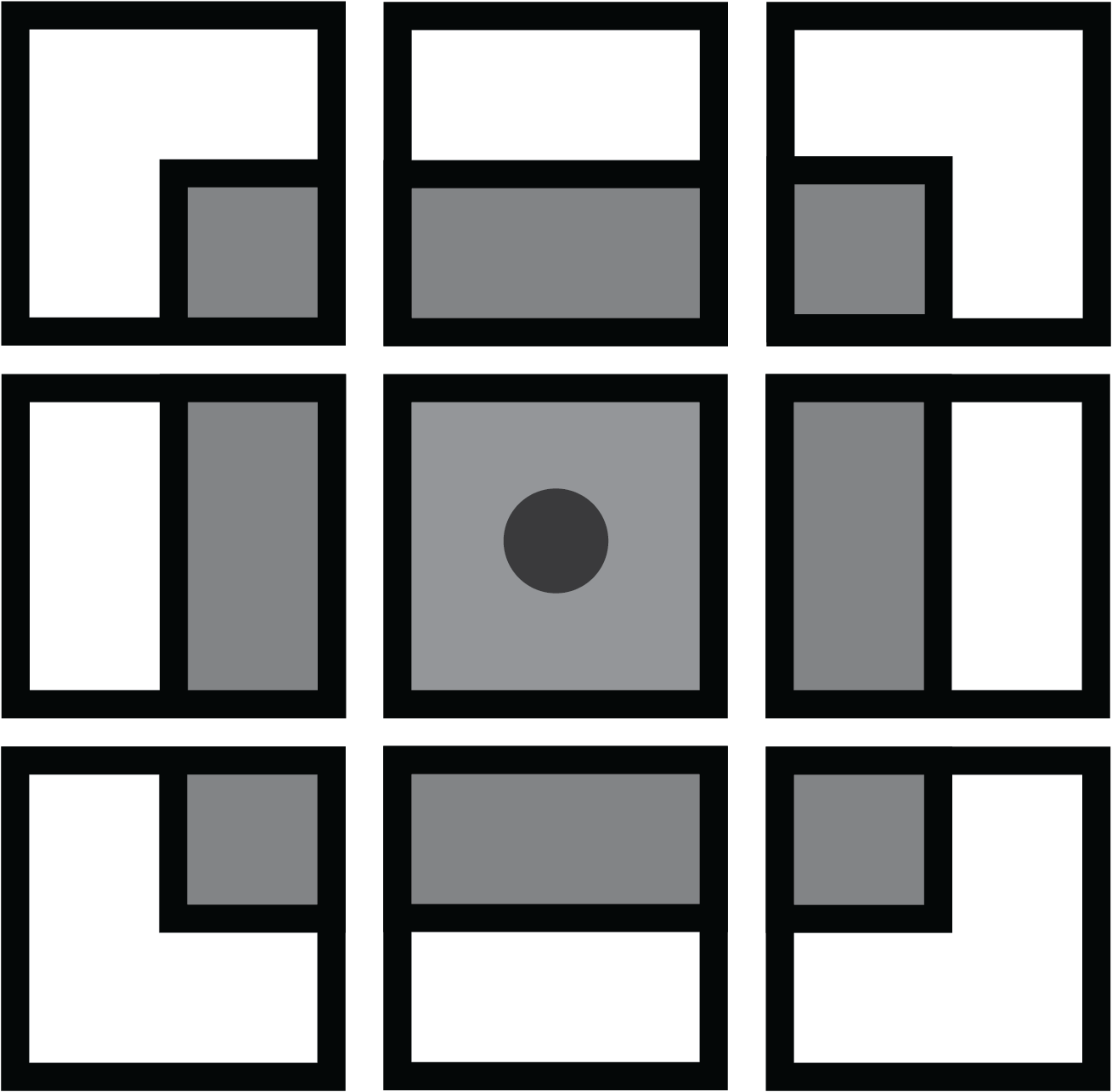}
    \caption{This is the only allowed configuration for the subtiles associated with each $3 \times 3$ grid. Each such grid represents a tile in the original tiling system.}
    \label{fig:3x3sim}
\end{figure}

Now we can incorporate an original rule between tiles by applying it to their corresponding border tiles. This will exactly simulate the original but with each site replaced by a $3\times3$ grid. This results in an $n \times n$ grid where $n = 3p$ and $p$ is a prime number.

With this tiling in hand, the 1D construction can now be embedded into the $3 \times 3$ grids associated with the $\horizontaltile$ and $\centertile$ tiles. In particular, only the middle row of the $3 \times 3$ grids associated with each $\horizontaltile$ and $\centertile$ tile is used for the chain (i.e. only the tiles $\lefttile$, $\middletile$, and $\righttile$). In addition, the $\middletile$ tile associated with the $\centertile$ tile is used to mark the left endpoint of the 1D construction.

It remains to construct a 1D translation-invariant Hamiltonian on a $f(x)$-length spin chain with the desired ground-state energy properties. To accomplish this \cite{v009a002} is directly used and is briefly summarized here. The main idea is to use the Hamiltonian terms to simulate a quantum Turing machine where each site in the length-$f(x)$ spin chain represents a different cell of the Turing machine tape except for two sites which are used to mark the boundaries. The goal of the first part of the Turing machine is to infer $x$ from the length of the tape and write it on the tape. The second part of the Turing machine is quantum and uses $x$ as the input along with a quantum witness to execute the $\qmaexp$ verification algorithm for $L$.

We now discuss the first part of the Turing machine, which we denote as $M_{BC}$. $M_{BC}$ is a purely classical Turing machine implementing a binary counter. By incrementing a clock pointer from one side of the tape to the other, we can ensure that $M_{BC}$ is run for exactly $f(x) - 3$ steps. Recall that $f(x) = 3p$ where $1/3$ of the most significant digits of $p$ is $x$. Therefore, by using a sufficiently slow binary counter, it is possible to ensure that $x$ is always written on the tape at the end of the $f(x) -3$ steps. 

Now that $x$ is on the tape, we can run the quantum verification algorithm on $x$ along with an arbitrary quantum witness. The verifier is also allowed a total of $f(x)-3$ timesteps. Notice this is more strict than $\qmaexp$, which allows the verifier $2^{\poly |x|}$. $\qmaexp$ can be reduced to this case by a standard padding argument where $x$ is padded by zeros to have length $\poly(|x|)$. Finally, an energy penalty is applied if the verifier does not accept. If the verifier accepts with probability $1-\epsilon$ in the $x \in L$ case then the ground-state energy will be upper bounded by $\epsilon/n^2$. In order to drive $\epsilon \leq O(n^{-k})$ for an arbitrarily large constant, it is possible to use witness amplification. This incurs a $O(k\log n)$ overhead in the verifier's runtime \cite{kitaev}, which can easily be accommodated by padding. In our construction, we would like to set $k = O(r)$ where $r$ is the dimension, but would prefer not to have the verification algorithm depend on $r$. Therefore, we can instead set $k=\log n$ where $\log n \geq O(r)$ for some $n\geq n_0$ since $r$ is a parameter of the problem and does not scale with $n$. Importantly, even though the number of rounds of witness amplification depends on $n$, our Hamiltonian term still does not depend on $n$ since $n$ is deduced from the length of the lattice and then given as input to the verification algorithm. For this construction \cite{v009a002} also show that in the $x \not\in L$ case the ground-state energy is lower bounded by $\Omega(1/n^3)$. 

Finally, each term in the construction is of one of two forms called Type 1 terms: $\ketbra{ab}{ab}$  and Type II terms: 
\begin{align}
    \frac{1}{2}(\ketbra{ab}{ab} + \ketbra{cd}{cd} - \ketbra{ab}{cd} - \ketbra{cd}{ab}).
\end{align}
Both types are positive semidefinite and so the overall Hamiltonian term is also positive semidefinite.

\subsection{Open boundary conditions}
\label{app:open}
An equivalent theorem to \cref{thm:GI} is also true in the case of open boundary conditions. The construction is also inspired by \cite{v009a002}. First, we define the tiles $\leftbc$, $\bottombc$, $\rightbc$, and $\blank$. The idea is then to introduce the following tiling rules: no tile is allowed to the left of or below $\leftbc$, no tile is allowed below $\bottombc$ and no tile is allowed to the right of or below $\rightbc$. Additionally, the only tile that is allowed to the right of $\leftbc$ is $\bottombc$ and the only tile that is allowed to the left of $\rightbc$ is also $\bottombc$. Finally, $\blank$ is not allowed to the left or right of $\bottombc$ and the only tile allowed above $\leftbc$, $\bottombc$ and $\rightbc$ is $\blank$. This results in the configuration depicted in \cref{fig:2dopenbc}. The 1D Hamiltonian can then be embedded into the bottom row of the 2D grid where the $\leftbc$ and $\rightbc$ tiles denote the endpoints.

\begin{figure}
    \centering
    \includegraphics[width=0.2\columnwidth]{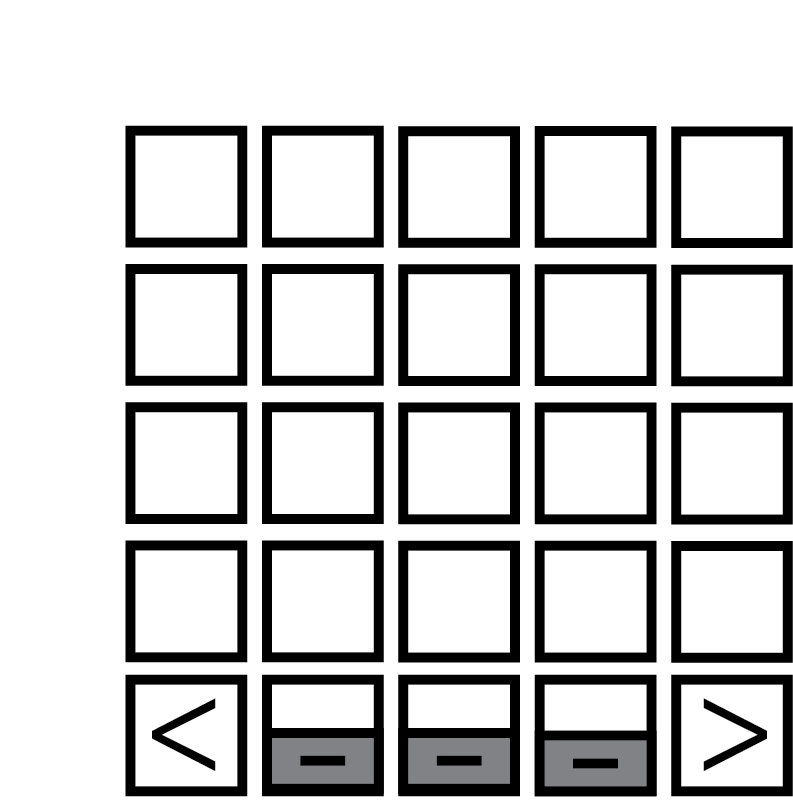}
    \caption{The only allowed configuration for the given set of translation-invariant tiling rules on a 2D grid with open boundary conditions.}
    \label{fig:2dopenbc}
\end{figure}
\end{document}